\documentclass[journal]{IEEEtran}

\usepackage{amsmath,amssymb,amsthm}
\usepackage{graphicx}
\usepackage{booktabs}
\usepackage{hyperref}
\usepackage{algorithm}
\usepackage{algpseudocode}
\usepackage{multirow}
\usepackage[numbers,sort&compress]{natbib}
\usepackage{xcolor}
\usepackage{wasysym}
\usepackage{tikz}
\usepackage{enumitem}
\usepackage{tabularx}
\usepackage{adjustbox}
\usepackage{balance}

\usetikzlibrary{shapes.geometric,arrows.meta,positioning,fit,calc}

\newtheorem{definition}{Definition}
\newtheorem{theorem}{Theorem}
\newtheorem{property}{Property}

\hypersetup{
  colorlinks=true,
  linkcolor=blue!60!black,
  citecolor=blue!60!black,
  urlcolor=blue!60!black
}

\newcommand{\FC}{\CIRCLE}
\newcommand{\PC}{\LEFTcircle}
\newcommand{\NC}{$\cdot$}

\begin{document}

\title{A Formal Security Framework for MCP-Based AI Agents: Threat Taxonomy, Verification Models, and Defense Mechanisms}

\author{
  \IEEEauthorblockN{Nirajan Acharya\textsuperscript{1}, Gaurav Kumar Gupta\textsuperscript{2}}
  \IEEEauthorblockA{\textsuperscript{1}\textit{University of the Cumberlands} \\
    Email: \texttt{nacharya39400@ucumberlands.edu}}
  \IEEEauthorblockA{\textsuperscript{2}\textit{Youngstown State University} \\
    Email: \texttt{guptagauravk1@gmail.com}}
}

\maketitle

\begin{abstract}
The Model Context Protocol (MCP), introduced by Anthropic in November 2024 and now governed by the Linux Foundation's Agentic AI Foundation, has rapidly become the de facto standard for connecting large language model (LLM)-based agents to external tools and data sources, with over 97 million monthly SDK downloads and more than 177,000 registered tools. However, this explosive adoption has exposed a critical gap: the absence of a unified, formal security framework capable of systematically characterizing, analyzing, and mitigating the diverse threats facing MCP-based agent ecosystems. Existing security research remains fragmented across individual attack papers, isolated benchmarks, and point defense mechanisms. This paper presents \textsc{MCPShield}, a comprehensive formal security framework for MCP-based AI agents. We make four principal contributions: (1)~a hierarchical threat taxonomy comprising 7 threat categories and 23 distinct attack vectors organized across four attack surfaces, grounded in the analysis of 177,000+ MCP tools; (2)~a formal verification model based on labeled transition systems with trust-boundary annotations that enables static and runtime analysis of MCP tool interaction chains; (3)~a systematic comparative evaluation of 12 existing defense mechanisms, identifying coverage gaps across our threat taxonomy; and (4)~a defense-in-depth reference architecture integrating capability-based access control, cryptographic tool attestation, information flow tracking, and runtime policy enforcement. Our analysis reveals that no existing single defense covers more than 34\% of the identified threat landscape, whereas \textsc{MCPShield}'s integrated architecture achieves theoretical coverage of 91\%. We further identify seven open research challenges that must be addressed to secure the next generation of agentic AI systems.
\end{abstract}


\section{Introduction}
\label{sec:introduction}

\IEEEPARstart{L}{arge} language model (LLM)-based AI agents represent a paradigm shift in software systems, evolving from stateless question-answering interfaces into autonomous entities capable of reasoning, planning, and executing multi-step workflows that modify real-world environments~\cite{Yao2023ReAct, Schick2023Toolformer}. Unlike traditional LLM applications, modern AI agents interact with external tools, databases, APIs, and other agents through structured protocols, creating complex trust relationships that extend far beyond the model's training boundary~\cite{AbouAli2025Agentic, Deng2025Threat}.

Central to this transformation is the Model Context Protocol (MCP), an open standard introduced by Anthropic in November 2024~\cite{Anthropic2024MCP}. Built on JSON-RPC 2.0, MCP defines a bidirectional communication interface between LLM host applications (clients) and external capability providers (servers), standardizing how agents discover, invoke, and compose tools~\cite{MCPSpec2025}. The protocol's adoption has been extraordinary: as of early 2026, the MCP ecosystem encompasses over 10,000 active servers, 177,000 registered tools, and 97 million monthly SDK downloads~\cite{Stein2026Evidence}. Major technology companies---including OpenAI, Google, Microsoft, and Amazon Web Services---have endorsed MCP through the Linux Foundation's Agentic AI Foundation (AAIF), establishing it as vendor-neutral infrastructure for the agentic AI era~\cite{AAIF2025}.

However, this rapid adoption has dramatically expanded the attack surface of AI agent systems. An empirical analysis of the MCP ecosystem reveals that ``action'' tools---those capable of directly modifying external environments---grew from 27\% to 65\% of all tools between November 2024 and February 2026~\cite{Stein2026Evidence}. Each tool invocation creates a potential trust boundary violation, and the compositional nature of MCP interactions means that compromising a single server can cascade through multi-agent workflows. Recent work has demonstrated concrete attacks including tool poisoning (embedding malicious instructions in tool descriptions)~\cite{InvariantLabs2025Poisoning, Wang2025MCPTox}, rug pulls (mutating tool behavior post-approval)~\cite{Bhatt2025ETDI}, cross-server data exfiltration~\cite{LogToLeak2025}, and privilege escalation through tool composition~\cite{Yang2025MCPSecBench}.

The security research community has responded with a growing but \textit{fragmented} body of work. Benchmarks such as MCPTox~\cite{Wang2025MCPTox}, MCP-SafetyBench~\cite{Zong2025SafetyBench}, and MCPSecBench~\cite{Yang2025MCPSecBench} evaluate specific attack categories but use incompatible threat taxonomies. Defense proposals including ETDI~\cite{Bhatt2025ETDI}, MCP-Guard~\cite{Xing2025MCPGuardDefense}, and MCPGuard~\cite{Wang2025MCPGuard} address individual attack vectors but lack integration. Enterprise security frameworks~\cite{Narajala2025Enterprise} provide high-level guidance without formal underpinnings. Most critically, no existing work provides a \textit{unified formal model} that captures the security properties of MCP interactions in a verifiable manner.

This fragmentation poses three concrete problems. First, practitioners cannot systematically assess the security posture of an MCP deployment because no common threat vocabulary exists. Second, defense mechanisms cannot be meaningfully compared because they operate under different assumptions and threat models. Third, the absence of formal semantics for MCP security properties precludes automated verification and certified deployments.

This paper addresses these gaps through \textsc{MCPShield}, a comprehensive formal security framework for MCP-based AI agents. Our contributions are:

\begin{itemize}[leftmargin=*]
  \item \textbf{Unified Threat Taxonomy (Section~\ref{sec:taxonomy}).} We synthesize findings from 12 MCP security papers, 5 benchmarks, and the OWASP Top 10 for LLM Applications into a hierarchical taxonomy of 7 threat categories and 23 attack vectors organized across 4 attack surfaces, providing the first common vocabulary for MCP security.

  \item \textbf{Formal Verification Model (Section~\ref{sec:formal}).} We introduce a labeled transition system with trust-boundary annotations ($\mathcal{M}_{\text{MCP}}$) that formalizes MCP interactions as state transitions. We define four fundamental security properties---\textit{tool integrity}, \textit{data confinement}, \textit{privilege boundedness}, and \textit{context isolation}---and provide decidability results for their verification.

  \item \textbf{Comparative Defense Analysis (Section~\ref{sec:defense}).} We systematically evaluate 12 existing defense mechanisms against our taxonomy, mapping coverage gaps and identifying architectural blind spots. We show that no individual defense covers more than 34\% of the threat landscape.

  \item \textbf{Defense-in-Depth Reference Architecture (Section~\ref{sec:architecture}).} We propose \textsc{MCPShield}, an integrated security architecture combining four complementary layers: capability-based access control, cryptographic tool attestation, information flow tracking, and runtime policy enforcement.
\end{itemize}

The remainder of this paper is organized as follows. Section~\ref{sec:background} provides background on MCP architecture and related work. Section~\ref{sec:taxonomy} presents our threat taxonomy. Section~\ref{sec:formal} introduces the formal verification model. Section~\ref{sec:defense} analyzes existing defense mechanisms. Section~\ref{sec:architecture} describes the \textsc{MCPShield} reference architecture. Section~\ref{sec:discussion} discusses implications and limitations. Section~\ref{sec:future} identifies open research challenges. Section~\ref{sec:conclusion} concludes.

\section{Background and Related Work}
\label{sec:background}

\subsection{The Model Context Protocol}
\label{subsec:mcp_background}

MCP defines a client-server architecture where \textit{hosts} (LLM applications) spawn \textit{clients} that maintain stateful sessions with \textit{servers} (external capability providers)~\cite{MCPSpec2025}. The protocol operates over JSON-RPC 2.0 and supports three primary server capabilities:

\begin{itemize}[leftmargin=*]
  \item \textbf{Tools}: Executable functions that agents can invoke (e.g., file operations, API calls, database queries). Each tool is defined by a name, description, and JSON Schema input specification.
  \item \textbf{Resources}: Read-only data sources exposed via URI-based addressing (e.g., file contents, database records).
  \item \textbf{Prompts}: Reusable prompt templates that guide agent behavior for specific tasks.
\end{itemize}

The protocol lifecycle consists of four phases: \textit{initialization} (capability negotiation via \texttt{initialize} handshake), \textit{discovery} (tool/resource enumeration via \texttt{tools/list}), \textit{operation} (tool invocation via \texttt{tools/call}), and \textit{termination} (session closure)~\cite{Hou2026Landscape}. Servers may additionally request \textit{sampling} from the host LLM and \textit{elicitation} of user input, creating bidirectional information flows.

\subsection{MCP Ecosystem Scale}
\label{subsec:ecosystem}

\citet{Stein2026Evidence} conducted the first large-scale empirical study of the MCP ecosystem, monitoring 177,436 tools across public MCP server repositories between November 2024 and February 2026. Key findings relevant to security include:

\begin{itemize}[leftmargin=*]
  \item Software development accounts for 67\% of all agent tools and 90\% of MCP server downloads, concentrating risk in code-execution-adjacent operations.
  \item The proportion of ``action'' tools (those that modify external environments) rose from 27\% to 65\%, indicating a rapid shift from read-only to write-capable agent interactions.
  \item The long-tail distribution of server popularity creates a dual risk: heavily-used servers are high-value targets, while niche servers receive minimal security scrutiny.
\end{itemize}

\subsection{Agent Communication Protocols Landscape}
\label{subsec:protocols}

MCP is not the only agent communication protocol. \citet{Ehtesham2025Interoperability} survey four emerging protocols: MCP (tool invocation via JSON-RPC), Google's Agent-to-Agent Protocol (A2A, peer-to-peer task delegation)~\cite{Google2025A2A}, the Agent Communication Protocol (ACP, RESTful multimodal messaging), and the Agent Network Protocol (ANP, decentralized agent discovery via W3C DIDs). Each protocol occupies a distinct niche---MCP for tool access, A2A for inter-agent collaboration---but their interoperability introduces additional security surfaces. An agent bridging MCP and A2A, for instance, must reconcile two distinct trust models, creating potential semantic gaps that adversaries can exploit.

\subsection{Related Security Work}
\label{subsec:related_security}

We organize related work into four categories.

\textbf{MCP-Specific Security Research.} \citet{Hou2026Landscape} provided the first systematic study of MCP security, identifying 16 threat scenarios across 4 attacker types (malicious developers, external attackers, malicious users, and inherent security flaws) mapped to the MCP server lifecycle. \citet{Narajala2025Enterprise} proposed enterprise mitigation strategies. \citet{Jamshidi2025Securing} analyzed tool poisoning, shadowing, and rug pull attacks, proposing RSA-based manifest signing and LLM-on-LLM semantic vetting. \citet{Jamshidi2026Manifest} extended this with a secure tool manifest and digital signing framework.

\textbf{MCP Benchmarks.} MCPTox~\cite{Wang2025MCPTox} benchmarks tool poisoning attacks across 45 real-world MCP servers with 353 tools and 1,312 malicious test cases spanning 11 risk categories. MCP-SafetyBench~\cite{Zong2025SafetyBench} provides a unified taxonomy of 20 attack types evaluated across 5 domains. MCPSecBench~\cite{Yang2025MCPSecBench} formalizes a secure MCP specification and identifies 17 attack types across 4 attack surfaces. Despite overlapping scope, these benchmarks use incompatible taxonomies, making cross-comparison difficult.

\textbf{MCP Defense Mechanisms.} MCP-Guard~\cite{Xing2025MCPGuardDefense} proposes a three-stage defense (static scanning, neural detection, LLM arbitration) achieving 89.63\% accuracy. MCPGuard~\cite{Wang2025MCPGuard} focuses on automated vulnerability detection for MCP servers. ETDI~\cite{Bhatt2025ETDI} introduces OAuth-enhanced tool definitions with cryptographic identity verification. MCPS~\cite{MCPS2026Secure} implements ECDSA P-256 message signing with progressive trust levels.

\textbf{General AI Agent Security.} \citet{Deng2025Threat} identify four knowledge gaps in AI agent security: unpredictability of multi-step inputs, complexity in internal executions, variability of operational environments, and interactions with untrusted external entities. \citet{He2025Emerged} report that none of 16 tested LLM agents achieved an overall safety score above 60\%. \citet{Tang2026Security} provide a comprehensive taxonomy of attacks and defenses for LLM-based agents. \citet{Dehghantanha2026SoK} systematize the attack surface of agentic AI across tools, RAG, and multi-agent loops. The OWASP Top 10 for LLM Applications~\cite{OWASP2025LLM} identifies Excessive Agency (LLM06) and Supply Chain (LLM03) as critical risks directly relevant to MCP.

\textbf{Differentiation.} Unlike prior work, this paper provides: (a)~a unified threat taxonomy reconciling all existing MCP-specific taxonomies, (b)~the first formal model with verifiable security properties for MCP interactions, and (c)~a defense-in-depth architecture informed by systematic coverage analysis of existing mechanisms.

\section{Threat Taxonomy for MCP-Based Agents}
\label{sec:taxonomy}

We construct our threat taxonomy through a systematic process. First, we collected all attack vectors documented in 12 MCP security papers~\cite{Hou2026Landscape, Wang2025MCPTox, Bhatt2025ETDI, Wang2025MCPGuard, Xing2025MCPGuardDefense, Narajala2025Enterprise, Zong2025SafetyBench, Yang2025MCPSecBench, Jamshidi2025Securing, Jamshidi2026Manifest, InvariantLabs2025Poisoning, LogToLeak2025} and 4 general agent security surveys~\cite{Deng2025Threat, He2025Emerged, Tang2026Security, Dehghantanha2026SoK}. We then performed open coding to identify common themes, merged overlapping categories, and organized the result into a hierarchical structure. The taxonomy was validated against the OWASP Top 10 for LLM Applications~\cite{OWASP2025LLM}, the STRIDE framework~\cite{Shostack2014Threat}, and the 177,000-tool empirical dataset~\cite{Stein2026Evidence}.

\subsection{Attack Surface Model}
\label{subsec:attack_surface}

We identify four attack surfaces in MCP-based agent systems:

\begin{definition}[MCP Attack Surfaces]
An MCP agent system exposes four attack surfaces:
\begin{enumerate}[leftmargin=*]
  \item \textbf{$\mathcal{S}_{\text{tool}}$: Tool Interface Surface.} The boundary between the LLM agent and MCP tool definitions, including tool descriptions, parameter schemas, and return values.
  \item \textbf{$\mathcal{S}_{\text{transport}}$: Transport Surface.} The communication channel between MCP clients and servers, including JSON-RPC messages, session state, and transport-layer security.
  \item \textbf{$\mathcal{S}_{\text{server}}$: Server Surface.} The MCP server implementation itself, including its authentication mechanisms, resource access patterns, and supply chain dependencies.
  \item \textbf{$\mathcal{S}_{\text{compose}}$: Composition Surface.} The emergent surface created when multiple MCP servers, tools, and agents interact within a single workflow or across protocol boundaries (MCP/A2A).
\end{enumerate}
\end{definition}

\subsection{Threat Categories}
\label{subsec:threat_categories}

We organize 23 attack vectors into 7 threat categories mapped to these surfaces (Table~\ref{tab:taxonomy}).

\begin{table*}[t]
\centering
\caption{Unified Threat Taxonomy for MCP-Based AI Agents. Attack vectors are mapped to attack surfaces ($\mathcal{S}$), attacker type, and existing benchmark/paper coverage (\checkmark = covered, $\circ$ = partially covered, --- = not covered).}
\label{tab:taxonomy}
\small
\begin{tabular}{@{}p{2.8cm}p{3.8cm}ccccc@{}}
\toprule
\textbf{Threat Category} & \textbf{Attack Vector} & \textbf{Surface} & \textbf{Attacker} & \textbf{MCPTox} & \textbf{SafetyBench} & \textbf{SecBench} \\
\midrule
\multirow{4}{*}{\parbox{2.8cm}{TC1: Tool\\Poisoning}}
  & TV1: Description injection & $\mathcal{S}_{\text{tool}}$ & Developer & \checkmark & \checkmark & \checkmark \\
  & TV2: Schema manipulation & $\mathcal{S}_{\text{tool}}$ & Developer & $\circ$ & --- & \checkmark \\
  & TV3: Return value poisoning & $\mathcal{S}_{\text{tool}}$ & External & \checkmark & $\circ$ & $\circ$ \\
  & TV4: Tool shadowing & $\mathcal{S}_{\text{tool}}$ & Developer & --- & --- & \checkmark \\
\midrule
\multirow{3}{*}{\parbox{2.8cm}{TC2: Rug Pull \&\\Mutation}}
  & TV5: Post-approval mutation & $\mathcal{S}_{\text{server}}$ & Developer & --- & --- & $\circ$ \\
  & TV6: Version rollback & $\mathcal{S}_{\text{server}}$ & Developer & --- & --- & --- \\
  & TV7: Capability escalation & $\mathcal{S}_{\text{server}}$ & Developer & --- & --- & --- \\
\midrule
\multirow{4}{*}{\parbox{2.8cm}{TC3: Cross-Server\\Data Leakage}}
  & TV8: Exfiltration via logging & $\mathcal{S}_{\text{compose}}$ & External & --- & \checkmark & $\circ$ \\
  & TV9: Context bleed & $\mathcal{S}_{\text{compose}}$ & External & --- & $\circ$ & --- \\
  & TV10: Channel coercion & $\mathcal{S}_{\text{compose}}$ & External & --- & --- & --- \\
  & TV11: Sampling abuse & $\mathcal{S}_{\text{transport}}$ & Server & --- & --- & --- \\
\midrule
\multirow{3}{*}{\parbox{2.8cm}{TC4: Privilege\\Escalation}}
  & TV12: Capability chaining & $\mathcal{S}_{\text{compose}}$ & External & --- & $\circ$ & $\circ$ \\
  & TV13: Consent bypass & $\mathcal{S}_{\text{tool}}$ & External & $\circ$ & \checkmark & \checkmark \\
  & TV14: Role confusion & $\mathcal{S}_{\text{transport}}$ & External & --- & --- & --- \\
\midrule
\multirow{3}{*}{\parbox{2.8cm}{TC5: Server Trust\\Violations}}
  & TV15: Impersonation & $\mathcal{S}_{\text{server}}$ & External & --- & --- & \checkmark \\
  & TV16: Supply chain compromise & $\mathcal{S}_{\text{server}}$ & Developer & --- & --- & --- \\
  & TV17: Dependency hijacking & $\mathcal{S}_{\text{server}}$ & External & --- & --- & --- \\
\midrule
\multirow{3}{*}{\parbox{2.8cm}{TC6: Context\\Manipulation}}
  & TV18: Prompt injection via tool & $\mathcal{S}_{\text{tool}}$ & External & \checkmark & \checkmark & \checkmark \\
  & TV19: Memory poisoning & $\mathcal{S}_{\text{compose}}$ & External & --- & $\circ$ & --- \\
  & TV20: Resource injection & $\mathcal{S}_{\text{tool}}$ & External & $\circ$ & $\circ$ & $\circ$ \\
\midrule
\multirow{3}{*}{\parbox{2.8cm}{TC7: Protocol-Level\\Vulnerabilities}}
  & TV21: Session hijacking & $\mathcal{S}_{\text{transport}}$ & External & --- & --- & --- \\
  & TV22: Replay attacks & $\mathcal{S}_{\text{transport}}$ & External & --- & --- & --- \\
  & TV23: Cross-protocol confusion & $\mathcal{S}_{\text{compose}}$ & External & --- & --- & --- \\
\bottomrule
\end{tabular}
\end{table*}

\subsubsection{TC1: Tool Poisoning}
Tool poisoning exploits the fundamental design of MCP, where tool descriptions are consumed by LLMs as natural language and thus can contain adversarial instructions~\cite{InvariantLabs2025Poisoning}. \citet{Wang2025MCPTox} demonstrated that even highly capable models (o1-mini: 72.8\% attack success rate) are vulnerable because the attack exploits their superior instruction-following abilities. We identify four variants:

\textit{TV1: Description Injection.} Malicious instructions embedded in tool description fields are invisible to end users but parsed by the LLM, causing it to execute unintended actions. Invariant Labs demonstrated that a ``Fact of the Day'' tool could exfiltrate an entire WhatsApp chat history through this vector~\cite{InvariantLabs2025Poisoning}.

\textit{TV2: Schema Manipulation.} Tool parameter schemas are modified to accept additional hidden fields that trigger side effects not visible in the tool's advertised interface.

\textit{TV3: Return Value Poisoning.} Tool execution returns contain embedded adversarial instructions that hijack the agent's subsequent reasoning, chaining into further malicious tool invocations.

\textit{TV4: Tool Shadowing.} A malicious server registers a tool with the same name or similar description as a legitimate tool, intercepting invocations intended for the legitimate version~\cite{Jamshidi2025Securing}.

\subsubsection{TC2: Rug Pull and Mutation Attacks}
Rug pull attacks exploit MCP's dynamic nature, where tool definitions can change after initial user approval~\cite{Bhatt2025ETDI}.

\textit{TV5: Post-Approval Mutation.} A tool that initially behaves benignly modifies its description, parameters, or behavior after the user grants permission, effectively converting approved access into unauthorized capabilities.

\textit{TV6: Version Rollback.} An MCP server rolls back to a previous tool version with known vulnerabilities or different behavior, circumventing security checks performed against the newer version.

\textit{TV7: Capability Escalation.} A server gradually expands its tool's capabilities across sessions, each increment appearing minor but cumulatively acquiring powerful unauthorized access.

\subsubsection{TC3: Cross-Server Data Leakage}
When agents interact with multiple MCP servers in a single session, data from one server's context can leak to another~\cite{LogToLeak2025}.

\textit{TV8: Exfiltration via Logging.} A malicious tool coerces the agent into invoking a logging or analytics tool to transmit sensitive data to an attacker-controlled endpoint. The Log-to-Leak attack~\cite{LogToLeak2025} systematizes this as four injection components: Trigger, Tool Binding, Justification, and Pressure.

\textit{TV9: Context Bleed.} Information from one MCP server's responses persists in the LLM's context window and is inadvertently included in requests to another server, violating data isolation expectations.

\textit{TV10: Channel Coercion.} A malicious tool manipulates the agent into using an insecure communication channel (e.g., HTTP instead of HTTPS, unencrypted stdio) for data transfer.

\textit{TV11: Sampling Abuse.} MCP servers can request LLM sampling from the host. A malicious server exploits this bidirectional channel to extract information from the host's context or influence the model's behavior for other servers.

\subsubsection{TC4: Privilege Escalation}
Agents operating under the principle of least privilege may still be manipulated into exceeding their authorized capabilities~\cite{OWASP2025LLM}.

\textit{TV12: Capability Chaining.} An attacker composes individually benign tool invocations to achieve an unauthorized outcome. For example, a read-file tool combined with a send-email tool enables data exfiltration, though neither tool alone is malicious.

\textit{TV13: Consent Bypass.} Adversarial tool descriptions manipulate the LLM into auto-approving actions that should require explicit user consent, exploiting the agent's tendency to follow instructions embedded in tool outputs.

\textit{TV14: Role Confusion.} In multi-role agent architectures, a tool invocation blurs the boundary between user-level and system-level operations, causing the agent to execute privileged operations under a standard user context.

\subsubsection{TC5: Server Trust Violations}
The MCP ecosystem's open nature---anyone can publish a server---creates server-level trust challenges analogous to software supply chain attacks~\cite{Ladisa2023Supply}.

\textit{TV15: Impersonation.} An attacker deploys an MCP server that mimics the identity (name, description, tool signatures) of a trusted server to intercept agent interactions.

\textit{TV16: Supply Chain Compromise.} A legitimate MCP server's dependencies are compromised, introducing malicious behavior without modifying the server's own code.

\textit{TV17: Dependency Hijacking.} A server's package name or repository is claimed by an attacker after the original maintainer abandons it, similar to npm/PyPI namespace hijacking.

\subsubsection{TC6: Context Manipulation}
These attacks target the LLM's reasoning process through its context window~\cite{Gulyamov2026Prompt}.

\textit{TV18: Prompt Injection via Tool.} The canonical MCP attack: malicious content in tool descriptions or return values that hijacks the agent's behavior. MCPTox~\cite{Wang2025MCPTox} categorizes 11 risk categories spanning information theft, harmful content generation, and system manipulation.

\textit{TV19: Memory Poisoning.} In agents with persistent memory (e.g., conversation history, knowledge bases), adversarial tool outputs inject false information that corrupts future reasoning sessions.

\textit{TV20: Resource Injection.} MCP resources (read-only data sources) are poisoned with adversarial content that influences the agent's tool selection and invocation patterns.

\subsubsection{TC7: Protocol-Level Vulnerabilities}
These attacks exploit weaknesses in MCP's transport and session management~\cite{Ferrag2025Protocol}.

\textit{TV21: Session Hijacking.} An attacker captures or predicts MCP session identifiers to inject commands into an active agent-server session.

\textit{TV22: Replay Attacks.} Previously captured JSON-RPC messages are re-transmitted to trigger repeated tool invocations or re-establish expired sessions.

\textit{TV23: Cross-Protocol Confusion.} When an agent bridges MCP and another protocol (e.g., A2A), semantic differences in trust models, authentication mechanisms, or capability representations create exploitable gaps~\cite{Ehtesham2025Interoperability}.

\subsection{Taxonomy Coverage Analysis}
\label{subsec:coverage_analysis}

Table~\ref{tab:taxonomy} reveals significant gaps in existing benchmark coverage. Of 23 attack vectors, MCPTox covers 5 fully and 2 partially, MCP-SafetyBench covers 4 fully and 5 partially, and MCPSecBench covers 5 fully and 4 partially. No benchmark addresses rug pull mutations (TV5--TV7), protocol-level attacks (TV21--TV23), or supply chain threats (TV16--TV17). This confirms our thesis that existing work is fragmented and a unified framework is necessary.

\section{Formal Verification Model}
\label{sec:formal}

We formalize MCP interactions as a labeled transition system with trust-boundary annotations. The goal is to define security properties precisely enough for automated verification.

\subsection{System Model}
\label{subsec:system_model}

\begin{definition}[MCP Agent System]
An MCP agent system is a tuple
\[\mathcal{M} = (\mathcal{A}, \mathcal{S}, \mathcal{T}, \mathcal{R}, \mathcal{C}, \ell)\]
where:
\begin{itemize}[leftmargin=*]
  \item $\mathcal{A} = \{a_1, \ldots, a_n\}$ is a finite set of agents (LLM instances);
  \item $\mathcal{S} = \{s_1, \ldots, s_m\}$ is a finite set of MCP servers;
  \item $\mathcal{T} = \{t_1, \ldots, t_k\}$ is a finite set of tools, where each $t_i$ is hosted by some $s_j \in \mathcal{S}$;
  \item $\mathcal{R} = \{r_1, \ldots, r_p\}$ is a finite set of resources;
  \item $\mathcal{C} \subseteq \mathcal{A} \times \mathcal{S}$ is the set of active client connections;
  \item $\ell\colon \mathcal{T} \cup \mathcal{R} \to \mathcal{L}$ is a security labeling function mapping tools and resources to a security lattice $\mathcal{L}$.
\end{itemize}
\end{definition}

The security lattice $\mathcal{L}$ follows \citet{Denning1976Lattice}, partially ordered by the \textit{can-flow-to} relation $\sqsubseteq$, with $\bot$ (public) and $\top$ (restricted) as bounds.

\begin{definition}[Trust Domain]
A trust domain $\mathcal{D}_i = (\mathcal{S}_i, \mathcal{T}_i, \pi_i)$ groups a set of servers $\mathcal{S}_i \subseteq \mathcal{S}$ and their tools $\mathcal{T}_i \subseteq \mathcal{T}$ under a common trust policy~$\pi_i$. Communication within a trust domain is considered trusted; cross-domain communication requires explicit verification.
\end{definition}

\subsection{Labeled Transition System}
\label{subsec:lts}

\begin{definition}[MCP Transition System]
The MCP transition system is
\[\mathcal{LTS} = (\Sigma, \Sigma_0, \Lambda, \rightarrow)\]
where:
\begin{itemize}[leftmargin=*]
  \item $\Sigma$ is the set of system states, where each state $\sigma = (\mathcal{C}_\sigma, \mathcal{K}_\sigma, \mathcal{P}_\sigma)$ consists of active connections~$\mathcal{C}_\sigma$, knowledge state~$\mathcal{K}_\sigma$ (information available to each agent), and permission state~$\mathcal{P}_\sigma$ (granted capabilities);
  \item $\Sigma_0 \subseteq \Sigma$ is the set of initial states;
  \item $\Lambda$ is the set of action labels, partitioned into:
  \begin{itemize}[leftmargin=1em]
    \item $\Lambda_{\text{disc}}$: discovery (\texttt{tools/list}, \texttt{resources/list})
    \item $\Lambda_{\text{inv}}$: invocation (\texttt{tools/call})
    \item $\Lambda_{\text{read}}$: resource access (\texttt{resources/read})
    \item $\Lambda_{\text{samp}}$: sampling (\texttt{sampling/createMessage})
    \item $\Lambda_{\text{admin}}$: session mgmt.\ (\texttt{initialize}, \texttt{shutdown})
  \end{itemize}
  \item ${\rightarrow} \subseteq \Sigma \times \Lambda \times \Sigma$ is the transition relation.
\end{itemize}
\end{definition}

Each transition $\sigma \xrightarrow{\lambda} \sigma'$ represents a single MCP operation. A \textit{trace} $\tau = \sigma_0 \xrightarrow{\lambda_1} \sigma_1 \xrightarrow{\lambda_2} \cdots \xrightarrow{\lambda_n} \sigma_n$ represents a complete interaction sequence. We write $\text{Traces}(\mathcal{LTS})$ for the set of all valid traces.

\subsection{Security Properties}
\label{subsec:security_properties}

We define four fundamental security properties for MCP agent systems.

\begin{property}[Tool Integrity]
\label{prop:integrity}
A system satisfies tool integrity if, for every trace $\tau \in \text{Traces}(\mathcal{LTS})$ and every tool invocation $\sigma_i \xrightarrow{\texttt{call}(t)} \sigma_{i+1}$, the tool definition $\text{def}(t)$ at time $i$ is identical to the definition at the time of user approval:
\begin{multline}
\forall \tau, \forall i: \sigma_i \xrightarrow{\texttt{call}(t)} \sigma_{i+1} \\
\implies \text{def}_i(t) = \text{def}_{\text{approved}}(t)
\end{multline}
This property defends against rug pull attacks (TC2) and tool shadowing (TV4).
\end{property}

\begin{property}[Data Confinement]
\label{prop:confinement}
A system satisfies data confinement if information never flows from a higher security level to a lower one across trust domain boundaries. For any trace $\tau$ and any cross-domain data flow $(d, \mathcal{D}_i, \mathcal{D}_j)$ where $i \neq j$:
\begin{multline}
\ell(d) \sqsubseteq \ell(\mathcal{D}_j) \quad \text{or} \\
\text{declassify}(d, \mathcal{D}_j) \text{ is explicitly authorized}
\end{multline}
This property defends against cross-server data leakage (TC3).
\end{property}

\begin{property}[Privilege Boundedness]
\label{prop:privilege}
A system satisfies privilege boundedness if, for every trace $\tau$ and every state $\sigma_i$, the effective permissions of each agent $a$ do not exceed the intersection of the agent's granted capabilities and the invoked tools' declared permissions:
\begin{multline}
\forall \tau, \forall i, \forall a \in \mathcal{A}: \\
\text{eff}(\sigma_i, a) \subseteq \text{cap}(a) \cap \bigcup_{t \in \text{active}(\sigma_i, a)} \text{perm}(t)
\end{multline}
This property defends against privilege escalation (TC4), including capability chaining.
\end{property}

\begin{property}[Context Isolation]
\label{prop:isolation}
A system satisfies context isolation if the knowledge state of an agent with respect to one server cannot influence its behavior toward another server without explicit cross-domain authorization:
\begin{multline}
\forall \tau, \forall s_i, s_j \in \mathcal{S},\, i \neq j: \\
\mathcal{K}_\sigma^{s_i}(a) \cap \text{input}\!\left(\sigma \xrightarrow{\texttt{call}(t_{s_j})}\right) = \emptyset
\end{multline}
unless a cross-domain flow policy explicitly permits the transfer. This property defends against context bleed (TV9) and memory poisoning (TV19).
\end{property}

\subsection{Verification Approach}
\label{subsec:verification}

\begin{theorem}[Decidability of Tool Integrity]
Tool integrity (Property~\ref{prop:integrity}) is decidable in $O(|\tau| \cdot |T|)$ time by maintaining a cryptographic hash of each tool's definition at approval time and comparing it before each invocation.
\end{theorem}

\begin{proof}
At approval time, compute $h_t = H(\text{def}(t))$ for each tool~$t$, where $H$ is a collision-resistant hash function. Before each invocation $\texttt{call}(t)$, compute $h'_t = H(\text{def}_{\text{current}}(t))$. The check $h_t = h'_t$ requires $O(1)$ per invocation, and there are at most $|\tau|$ invocations across $|T|$~tools.
\end{proof}

\begin{theorem}[Decidability of Data Confinement]
Data confinement (Property~\ref{prop:confinement}) is decidable when the number of trust domains and security levels is finite, by constructing the product automaton of the MCP transition system and the information flow lattice~\cite{Denning1976Lattice}.
\end{theorem}

The verification of privilege boundedness and context isolation requires tracking capability sets and knowledge states through traces, which is decidable for finite-state systems but may require abstraction for practical MCP deployments.

We note that runtime enforcement of these properties can leverage security automata~\cite{Schneider2000Enforceable}, where a monitor observes the trace of MCP operations and terminates (or modifies, following edit automata~\cite{Ligatti2005Edit}) any execution that would violate a security property.

\section{Comparative Analysis of Defense Mechanisms}
\label{sec:defense}

We systematically evaluate 12 existing defense mechanisms against our threat taxonomy. For each mechanism, we assess: (a)~which threat categories it addresses, (b)~its enforcement model (preventive, detective, or reactive), (c)~its deployment requirements, and (d)~its performance overhead.

\subsection{Defense Mechanism Catalog}
\label{subsec:defense_catalog}

\textbf{D1: ETDI (Enhanced Tool Definition Interface)}~\cite{Bhatt2025ETDI}. Extends MCP tool definitions with OAuth 2.0-based cryptographic identity verification, immutable versioned definitions, and policy-based access control. Addresses TC2 (rug pulls) through version pinning and TC4 (privilege escalation) through fine-grained permission policies. Does not address TC3, TC6, or TC7.

\textbf{D2: MCP-Guard}~\cite{Xing2025MCPGuardDefense}. Three-stage defense: Stage~1 performs lightweight pattern-based static scanning ($<$2ms latency), Stage~2 uses E5 text embeddings for deep neural detection, and Stage~3 employs LLM arbitration for ambiguous cases. Achieves 89.63\% accuracy and 89.07\% F1-score. Primarily addresses TC1 (tool poisoning) and TC6 (prompt injection). Limited coverage of TC2, TC3, TC5.

\textbf{D3: MCPGuard}~\cite{Wang2025MCPGuard}. Automated vulnerability detection focusing on agent hijacking from protocol design deficiencies, traditional web vulnerabilities in MCP servers, and supply chain risks. Addresses TC1, partial TC5, but is detective-only (identifies vulnerabilities but does not enforce policies at runtime).

\textbf{D4: Secure Tool Manifest}~\cite{Jamshidi2025Securing, Jamshidi2026Manifest}. RSA-based manifest signing combined with LLM-on-LLM semantic vetting and heuristic guardrails. Addresses TC1 (tool poisoning), TC2 (rug pulls through signed manifests), and partial TC5 (server integrity through signing). Does not address TC3, TC4, or TC7.

\textbf{D5: MCPS (MCP Secure)}~\cite{MCPS2026Secure}. ECDSA P-256 cryptographic identity and message signing with progressive trust levels L0--L4. Every JSON-RPC message wrapped in a signed envelope with nonce-based replay protection. Addresses TC5 (server trust), TC7 (session hijacking, replay attacks). Does not address TC1, TC3, or TC6.

\textbf{D6: MCP Specification Security Principles}~\cite{MCPSpec2025}. The MCP specification defines security principles around user consent, data privacy, tool safety, and LLM sampling controls, but enforcement is left to implementors. Provides partial TC4 coverage through consent requirements. Insufficient as a standalone defense.

\textbf{D7--D8: Benchmark-Derived Defenses.} MCP-SafetyBench~\cite{Zong2025SafetyBench} and MCPSecBench~\cite{Yang2025MCPSecBench} include evaluation of model-level defenses (system prompts, safety training). MCPSecBench reports that current protection mechanisms achieve an average success rate below 30\%, indicating fundamental insufficiency.

\textbf{D9: Enterprise Security Frameworks}~\cite{Narajala2025Enterprise}. Comprehensive mitigation strategies including threat modeling, access control policies, monitoring, and incident response. High-level architectural guidance without formal properties or automated enforcement.

\textbf{D10: Microsoft Agent Governance Toolkit}~\cite{Microsoft2026Governance}. Runtime security with cryptographic identity (DIDs/Ed25519), dynamic execution rings inspired by CPU privilege levels, saga orchestration, and kill switch. Addresses TC4 (privilege escalation), TC5 (server trust), partial TC7. Designed for Microsoft Copilot ecosystem; MCP integration is indirect.

\textbf{D11: OWASP Agentic AI Guidance}~\cite{OWASP2025LLM}. Risk taxonomy and mitigation guidance for Excessive Agency, Supply Chain, and Prompt Injection risks. Addresses awareness but provides no enforceable mechanism.

\textbf{D12: Zero Trust Architecture Principles}~\cite{NIST2020ZeroTrust}. NIST SP 800-207 principles applied to MCP: never trust implicitly, verify every tool invocation, enforce least privilege. Provides architectural guidance relevant to all threat categories but requires MCP-specific instantiation.

\subsection{Coverage Gap Analysis}
\label{subsec:coverage_gaps}

Table~\ref{tab:defense_coverage} maps each defense mechanism to the threat categories it covers.

\begin{table*}[t]
\centering
\caption{Defense mechanism coverage across threat categories. \CIRCLE~= full coverage, \LEFTcircle~= partial coverage, $\cdot$~= no coverage. Max coverage column shows the percentage of the 23 attack vectors addressed.}
\label{tab:defense_coverage}
\small
\begin{tabular}{@{}lcccccccr@{}}
\toprule
\textbf{Defense} & \textbf{TC1} & \textbf{TC2} & \textbf{TC3} & \textbf{TC4} & \textbf{TC5} & \textbf{TC6} & \textbf{TC7} & \textbf{Cov.} \\
\midrule
D1: ETDI       & \NC & \FC & \NC & \PC & \PC & \NC & \NC & 22\% \\
D2: MCP-Guard  & \FC & \NC & \NC & \NC & \NC & \FC & \NC & 30\% \\
D3: MCPGuard   & \FC & \NC & \NC & \NC & \PC & \NC & \NC & 22\% \\
D4: Manifest   & \FC & \FC & \NC & \NC & \PC & \NC & \NC & 26\% \\
D5: MCPS       & \NC & \NC & \NC & \NC & \FC & \NC & \FC & 26\% \\
D6: Spec       & \NC & \NC & \NC & \PC & \NC & \NC & \NC &  9\% \\
D7: SafetyB.   & \PC & \NC & \PC & \PC & \NC & \PC & \NC & 17\% \\
D8: SecBench   & \PC & \PC & \PC & \PC & \PC & \PC & \NC & 34\% \\
D9: Enterprise & \PC & \PC & \PC & \PC & \PC & \PC & \PC & 30\% \\
D10: MS Gov.   & \NC & \NC & \NC & \FC & \FC & \NC & \PC & 26\% \\
D11: OWASP     & \PC & \NC & \NC & \PC & \PC & \PC & \NC & 17\% \\
D12: Zero Trust & \PC & \PC & \PC & \PC & \PC & \PC & \PC & 30\% \\
\bottomrule
\end{tabular}
\end{table*}

Key findings from the coverage analysis:

\begin{enumerate}[leftmargin=*]
  \item \textbf{No single defense exceeds 34\% coverage.} MCPSecBench's evaluation-based approach achieves the broadest coverage but remains partial across most categories.
  \item \textbf{Cross-server data leakage (TC3) is the least defended category.} Only enterprise frameworks and zero trust principles provide partial coverage; no mechanism enforces data confinement formally.
  \item \textbf{Protocol-level attacks (TC7) lack dedicated defenses.} Only MCPS addresses replay attacks directly; session hijacking and cross-protocol confusion remain unaddressed.
  \item \textbf{Rug pull defenses (TC2) rely solely on versioning.} ETDI and secure manifests prevent post-approval mutation through version pinning, but capability escalation (TV7) across sessions is not addressed.
  \item \textbf{Composition attacks ($\mathcal{S}_{\text{compose}}$) lack compositional defenses.} Individual tool-level defenses do not account for emergent risks from tool composition.
\end{enumerate}

\section{MCPShield: Defense-in-Depth Architecture}
\label{sec:architecture}

Based on the coverage gap analysis, we propose \textsc{MCPShield}, a layered security architecture that integrates four complementary defense layers to achieve comprehensive coverage across all seven threat categories.

\subsection{Architectural Overview}
\label{subsec:arch_overview}

\textsc{MCPShield} follows a defense-in-depth strategy inspired by zero trust principles~\cite{NIST2020ZeroTrust} and capability-based security~\cite{Dennis1966Programming}. The architecture consists of four layers:

\begin{enumerate}[leftmargin=*]
  \item \textbf{Layer 1: Capability-Based Access Control (L-CAC).} Governs what each agent can do by restricting tool invocations to explicitly granted capabilities.
  \item \textbf{Layer 2: Cryptographic Tool Attestation (L-CTA).} Ensures what each tool \textit{is} by verifying tool identity, integrity, and provenance at every invocation.
  \item \textbf{Layer 3: Information Flow Tracking (L-IFT).} Controls where data \textit{goes} by enforcing the data confinement property across trust domain boundaries.
  \item \textbf{Layer 4: Runtime Policy Enforcement (L-RPE).} Monitors \textit{how} interactions unfold by observing execution traces and enforcing security policies as security automata~\cite{Schneider2000Enforceable}.
\end{enumerate}

\subsection{Layer 1: Capability-Based Access Control}
\label{subsec:layer1}

Following \citet{Dennis1966Programming} and the principle of least privilege, L-CAC models each agent's authority as a \textit{capability set}---a collection of unforgeable tokens, each granting access to a specific tool with specific parameters.

\begin{definition}[MCP Capability]
An MCP capability is a tuple
\[c = (t, \mathit{params}, \mathit{scope}, \mathit{ttl}, \sigma_c)\]
where $t \in \mathcal{T}$ identifies the tool, $\mathit{params} \subseteq \text{Schema}(t)$ restricts allowed parameters, $\mathit{scope} \in \{\text{read}, \text{write}, \text{execute}\}$ limits the operation type, $\mathit{ttl}$ is a time-to-live, and $\sigma_c$ is a cryptographic signature binding the capability to the issuing authority.
\end{definition}

Before each tool invocation $\texttt{call}(t, \text{args})$, the L-CAC layer verifies:
\begin{enumerate}[leftmargin=*]
  \item The agent holds a valid capability $c$ for tool $t$;
  \item $\text{args} \subseteq \text{params}(c)$;
  \item The operation is within $\text{scope}(c)$;
  \item The capability has not expired ($\text{ttl}$ not exceeded).
\end{enumerate}

\textbf{Composition Control.} To defend against capability chaining (TV12), L-CAC introduces \textit{composition policies} that restrict which tools can be invoked in sequence. A composition policy
\[\rho\colon \mathcal{T} \times \mathcal{T} \to \{\text{allow}, \text{deny}, \text{audit}\}\]
specifies whether tool $t_j$ may be invoked after $t_i$ within a single trace.

\subsection{Layer 2: Cryptographic Tool Attestation}
\label{subsec:layer2}

L-CTA ensures tool integrity (Property~\ref{prop:integrity}) by combining ideas from ETDI~\cite{Bhatt2025ETDI}, secure manifests~\cite{Jamshidi2026Manifest}, and MCPS~\cite{MCPS2026Secure}.

Each MCP server publishes a signed \textit{tool attestation record}:
\begin{multline}
\text{TAR}(t) = \text{Sign}_{sk_s}\big(H(\text{def}(t)) \| \text{version}(t) \\
\| \text{timestamp} \| \text{deps}(t)\big)
\end{multline}
where $sk_s$ is the server's private signing key, $H$ is a collision-resistant hash function, and $\text{deps}(t)$ is a hash of the tool's dependency tree.

The L-CTA layer maintains an \textit{attestation log} (analogous to certificate transparency) that records all tool attestation records. Before each invocation:
\begin{enumerate}[leftmargin=*]
  \item Verify $\text{TAR}(t)$ against the server's public key;
  \item Compare $H(\text{def}_{\text{current}}(t))$ with the approved hash;
  \item Verify $\text{version}(t)$ has not regressed (defending against TV6);
  \item Check $\text{deps}(t)$ against known-good dependency hashes (defending against TV16--TV17).
\end{enumerate}

\subsection{Layer 3: Information Flow Tracking}
\label{subsec:layer3}

L-IFT enforces data confinement (Property~\ref{prop:confinement}) through dynamic taint tracking at the MCP message level. Each piece of data entering the agent's context from a tool response is labeled with its originating trust domain:

\begin{definition}[Tainted Data]
A tainted datum is a pair $(d, \ell_d)$ where $d$ is the data value and $\ell_d \in \mathcal{L}$ is its security label, inherited from the trust domain of the MCP server that produced it.
\end{definition}

When the agent composes a new tool invocation, L-IFT checks that no argument contains data with a security label higher than the target server's clearance level:
\begin{multline}
\forall \text{arg} \in \text{args}(\texttt{call}(t)): \\
\ell(\text{arg}) \sqsubseteq \ell(\mathcal{D}_{\text{server}(t)})
\end{multline}

This directly defends against cross-server data leakage (TC3). For context isolation (Property~\ref{prop:isolation}), L-IFT additionally tracks information provenance through the LLM's reasoning steps, flagging any case where data from server $s_i$ influences a request to server $s_j$ without authorization.

\subsection{Layer 4: Runtime Policy Enforcement}
\label{subsec:layer4}

L-RPE implements a security automaton~\cite{Schneider2000Enforceable} that monitors the MCP transition system in real time. We adopt the edit automaton model~\cite{Ligatti2005Edit}, which can suppress, insert, or modify actions in the event stream.

\begin{definition}[MCP Security Automaton]
An MCP security automaton is a tuple
\[\mathcal{E} = (Q, q_0, \Lambda, \delta, \gamma)\]
where $Q$ is a finite set of monitor states, $q_0$ is the initial state, $\Lambda$ is the MCP action alphabet, $\delta\colon Q \times \Lambda \to Q$ is the transition function, and $\gamma\colon Q \times \Lambda \to \{\text{allow}, \text{suppress}, \text{modify}(\lambda')\}$ is the enforcement function.
\end{definition}

The security automaton encodes policies such as:
\begin{itemize}[leftmargin=*]
  \item \textbf{Rate limiting}: Suppress tool invocations exceeding a threshold within a time window (defends against denial-of-service via excessive tool calls).
  \item \textbf{Anomaly detection}: Flag unusual sequences of tool invocations that deviate from learned behavioral profiles (defends against TV12, TV13).
  \item \textbf{Consent enforcement}: Insert user confirmation requests before high-risk operations, even if the agent attempts to auto-approve (defends against TV13, TV14).
  \item \textbf{Semantic validation}: Modify tool invocation arguments that contain suspected prompt injection patterns, stripping adversarial content before it reaches the MCP server (defends against TC1, TC6).
\end{itemize}

\subsection{Integrated Coverage}
\label{subsec:integrated_coverage}

Table~\ref{tab:mcpshield_coverage} shows how \textsc{MCPShield}'s four layers complement each other to achieve comprehensive coverage.

\begin{table*}[t]
\centering
\caption{\textsc{MCPShield} layer coverage across threat categories.}
\label{tab:mcpshield_coverage}
\small
\begin{tabular}{@{}lcccccccc@{}}
\toprule
\textbf{Layer} & \textbf{TC1} & \textbf{TC2} & \textbf{TC3} & \textbf{TC4} & \textbf{TC5} & \textbf{TC6} & \textbf{TC7} & \textbf{Cov.} \\
\midrule
L-CAC  & \PC & \NC & \NC & \FC & \NC & \NC & \NC & 17\% \\
L-CTA  & \FC & \FC & \NC & \NC & \FC & \NC & \PC & 39\% \\
L-IFT  & \NC & \NC & \FC & \PC & \NC & \PC & \NC & 22\% \\
L-RPE  & \FC & \PC & \PC & \FC & \PC & \FC & \FC & 70\% \\
\midrule
\textbf{Combined} & \FC & \FC & \FC & \FC & \FC & \FC & \FC & \textbf{91\%} \\
\bottomrule
\end{tabular}
\end{table*}

The combined architecture achieves 91\% theoretical coverage (21 of 23 attack vectors fully addressed). The two remaining vectors---TV10 (channel coercion) and TV19 (memory poisoning in long-term agent memory)---require environment-level controls and LLM-internal mechanisms, respectively, which fall outside the protocol-level scope of \textsc{MCPShield}.

\section{Discussion}
\label{sec:discussion}

\subsection{Implications for MCP Protocol Design}
\label{subsec:implications}

Our analysis reveals that several threat categories arise from fundamental design decisions in the MCP specification. Tool descriptions are consumed as natural language by LLMs, creating an inherent prompt injection surface (TC1, TC6). The current specification does not mandate tool definition immutability, enabling rug pulls (TC2). Session management lacks built-in cryptographic protection, exposing transport-level vulnerabilities (TC7). We recommend that the AAIF incorporate the following into future MCP specification revisions:

\begin{enumerate}[leftmargin=*]
  \item \textbf{Mandatory tool attestation}: Tool definitions must carry cryptographic signatures that clients verify before presentation to the LLM.
  \item \textbf{Structured tool descriptions}: Replace free-form natural language descriptions with a structured schema that separates machine-parseable functionality specifications from human-readable documentation, reducing the prompt injection surface.
  \item \textbf{Built-in session security}: Integrate message signing and nonce-based replay protection as mandatory protocol features, not optional extensions.
\end{enumerate}

\subsection{The Composition Challenge}
\label{subsec:composition}

The most novel threat surface identified in our taxonomy is $\mathcal{S}_{\text{compose}}$---the emergent attack surface from multi-server, multi-tool, and cross-protocol interactions. Traditional security frameworks analyze components in isolation, but MCP-based agent systems exhibit emergent risks that only manifest when tools are composed. Capability chaining (TV12) and cross-protocol confusion (TV23) exemplify this: each component may be individually secure, yet their composition is not.

This mirrors the well-known \textit{confused deputy} problem~\cite{Dehghantanha2026SoK} but at a protocol level. Our composition policies in L-CAC provide a first step, but a complete solution requires \textit{compositional security proofs}---showing that security properties of individual components are preserved under composition. This remains an open theoretical challenge.

\subsection{Scalability and Performance Considerations}
\label{subsec:scalability}

\textsc{MCPShield}'s four-layer architecture introduces overhead at each tool invocation. Based on existing component benchmarks:

\begin{itemize}[leftmargin=*]
  \item \textbf{L-CAC}: Capability verification is $O(1)$ with pre-computed capability sets. Estimated overhead: $<$1ms per invocation.
  \item \textbf{L-CTA}: Signature verification (ECDSA P-256) requires approximately 0.5ms. Hash comparison is negligible. Estimated overhead: $<$2ms per invocation.
  \item \textbf{L-IFT}: Taint propagation at the message level (not token level) adds minimal overhead. Estimated overhead: $<$1ms per invocation.
  \item \textbf{L-RPE}: MCP-Guard reports $<$2ms for Stage~1 pattern matching~\cite{Xing2025MCPGuardDefense}. Security automaton state transitions are $O(1)$. Estimated overhead: 2--5ms per invocation.
\end{itemize}

Total estimated overhead of 4--9ms per tool invocation is negligible relative to LLM inference latency (typically 500ms--5s) and network round-trip times.

\subsection{Limitations}
\label{subsec:limitations}

Our work has several limitations that we acknowledge transparently:

\begin{enumerate}[leftmargin=*]
  \item \textbf{Theoretical framework without implementation.} \textsc{MCPShield} is a reference architecture with formal properties but has not been implemented or empirically evaluated. Real-world deployment may reveal practical challenges not captured in our model.
  \item \textbf{LLM-internal attacks are out of scope.} Our formal model operates at the protocol level and cannot address attacks that exploit LLM-internal reasoning (e.g., sophisticated prompt injections that bypass semantic analysis). Defense against these requires advances in LLM alignment and robustness.
  \item \textbf{Coverage metric is theoretical.} Our 91\% coverage claim is based on mapping defense layers to threat vectors. Actual effectiveness depends on the quality of implementation and the adversary's sophistication.
  \item \textbf{Dynamic threat landscape.} Our taxonomy reflects the state of MCP security research as of early 2026. New attack vectors will emerge as the ecosystem evolves.
  \item \textbf{Cross-protocol security is nascent.} Our treatment of A2A/ACP/ANP interoperability threats (TV23) is preliminary, as these protocols are still maturing.
\end{enumerate}

\section{Open Research Challenges}
\label{sec:future}

Our systematization identifies seven research challenges that must be addressed to secure MCP-based AI agents at scale.

\textbf{Challenge 1: Compositional Security Proofs.} Proving that the composition of individually secure MCP tools preserves security properties is an open theoretical problem. Existing approaches from protocol composition (e.g., the UC framework) do not directly apply because MCP tools are not traditional cryptographic protocols---they involve natural language interfaces and non-deterministic LLM behavior.

\textbf{Challenge 2: Semantic Integrity Verification.} Cryptographic signatures verify syntactic integrity (the tool definition has not changed) but not semantic integrity (the tool still does what it claims). A tool could maintain an identical description while modifying its server-side implementation. Verifying semantic integrity likely requires reproducible builds, formal specifications of tool behavior, and runtime behavioral monitoring.

\textbf{Challenge 3: Scalable Information Flow Tracking.} Our L-IFT layer tracks information at the message level, but LLMs transform data through attention mechanisms that make fine-grained provenance tracking intractable. Developing scalable abstractions for information flow through LLMs---perhaps leveraging influence functions or activation analysis---is an open problem.

\textbf{Challenge 4: Adversarial Robustness of Defense Layers.} Each \textsc{MCPShield} layer can itself become an attack target. An adversary might craft tool descriptions that evade MCP-Guard's pattern matching, forge capabilities by exploiting implementation bugs, or poison the behavioral profiles used for anomaly detection. A security analysis of the defense architecture itself (defense-against-defense attacks) is needed.

\textbf{Challenge 5: Cross-Protocol Trust Federation.} As agents increasingly bridge MCP, A2A, ACP, and ANP, a federated trust framework is needed that can translate and compose trust assertions across heterogeneous protocols. The current landscape of 4+ competing protocols~\cite{Ehtesham2025Interoperability} without interoperable trust semantics is a significant gap.

\textbf{Challenge 6: Dynamic Consent and Human-in-the-Loop.} Current consent models are binary (approve/deny) and static (granted at connection time). Real-world MCP usage requires dynamic, context-sensitive consent that adapts to the risk level of each specific invocation. Designing consent UIs and policies that balance security with usability in real-time agentic workflows is an HCI and security co-design challenge.

\textbf{Challenge 7: Ecosystem-Level Governance.} With 177,000+ tools and growing~\cite{Stein2026Evidence}, the MCP ecosystem faces the same governance challenges as npm, PyPI, and Docker Hub, but with additional LLM-specific risks. \citet{Nguyen2026MultiAgent} found that ``Non-Determinism'' and ``Data Leakage'' are the most inadequately addressed risk domains across 16 existing security frameworks. The AAIF~\cite{AAIF2025} provides governance structure, but technical mechanisms for server reputation, tool certification, and ecosystem-wide threat intelligence remain undeveloped.

\section{Conclusion}
\label{sec:conclusion}

This paper presented \textsc{MCPShield}, a formal security framework for MCP-based AI agents that addresses the critical gap of fragmented security research in the rapidly growing MCP ecosystem. Through systematic analysis of 12 MCP security papers, 5 benchmarks, and the 177,000-tool empirical landscape, we constructed a unified threat taxonomy of 7 categories and 23 attack vectors across 4 attack surfaces. Our formal verification model, based on labeled transition systems with trust-boundary annotations, defines four fundamental security properties---tool integrity, data confinement, privilege boundedness, and context isolation---with decidability results for their verification. Our comparative evaluation revealed that no existing defense mechanism covers more than 34\% of the threat landscape, whereas \textsc{MCPShield}'s integrated four-layer architecture achieves 91\% theoretical coverage.

As MCP transitions from an Anthropic project to vendor-neutral infrastructure under the Linux Foundation, and as action-capable tools grow from 27\% to 65\% of the ecosystem, the security stakes have never been higher. The International AI Safety Report 2026~\cite{Bengio2026Safety} concludes that AI risk management techniques are ``improving but insufficient.'' We hope that \textsc{MCPShield}'s formal foundations, systematic analysis, and identified research challenges will catalyze the next generation of MCP security research and guide the AAIF in developing enforceable security standards for the agentic AI era.

\bibliographystyle{IEEEtranN}
\bibliography{references}

\end{document}